\documentclass[12pt]{article}
\usepackage{amsmath,amssymb,amsthm,graphicx,mathrsfs,color,fancyhdr,fancybox,setspace}
\usepackage{placeins}
\usepackage{amsfonts}

\newtheorem{theorem}{Theorem}[section]

\newtheorem{conjecture}[theorem]{Conjecture}

  \newtheorem{observation}[theorem]{Observation}

\newtheorem{lemma}[theorem]{Lemma}

\newcommand{\prooff}[1]{\textit{Proof of Theorem #1.}}

\begin{document}
\thispagestyle{empty}
\begin{center}
\Large
 {\bf Vertex removal in biclique graphs}
\vspace*{1cm}

\large
Leandro Montero\\
\vspace*{0.5cm}
\normalsize
KLaIM team, L@bisen, Yncrea Ouest\\33 Q, Chemin du Champ de Man\oe{}uvres\\44470 Carquefou, France\\
  \texttt{lpmontero@gmail.com} \\
\vspace*{.5cm}

ABSTRACT
\end{center}

\small A \textit{biclique} is a maximal induced complete bipartite subgraph.
The \textit{biclique graph} of a graph $H$, denoted by $KB(H)$, is the intersection graph of the family of all bicliques of $H$.
In this work we address the following question: Given a biclique graph $G=KB(H)$, is it possible to remove a vertex $q$ of $G$, such that
$G - \{q\}$ is a biclique graph? And if possible, can we obtain a graph $H'$ such that $G - \{q\} = KB(H')$? We show that the general
question has a ``no'' for answer. However, we prove that if $G$ has a vertex $q$ such that $d(q) = 2$, then $G-\{q\}$ is a biclique graph and we
show how to obtain $H'$.

\normalsize
\vspace*{1cm}

\section{Introduction}

Given a family of sets $\mathcal{H}$, the \textit{intersection graph} of $\mathcal{H}$ is a graph that has the members of 
$\mathcal{H}$ as vertices and there is an edge between two sets $E,F\in\mathcal{H}$ when $E$ and $F$ have non-empty intersection.
A graph $G$ is an \textit{intersection graph} of a specific type of objects if there exists a family of sets $\mathcal{H}$ of these objects such that $G$ is the intersection graph of $\mathcal{H}$. We remark that every graph is an intersection graph of appropriately defined objects~\cite{Szpilrajn-MarczewskiFM1945}. Intersection graphs of certain special subgraphs of a general graph have been studied 
extensively. 

A \textit{clique} of $G$ is a maximal complete subgraph.
The \textit{clique graph} of $G$, denoted by $K(G)$, is the intersection graph of the family of all cliques of $G$.
Clique graphs were introduced in~\cite{HamelinkJCT1968} and characterized in~\cite{RobertsSpencerJCTSB1971}. It was proved in~\cite{Alc'onFariaFigueiredoGutierrez2006} 
that the clique graph recognition problem is $NP-$complete.

A \textit{biclique} of $G$ is a maximal induced complete bipartite subgraph. 
The \textit{biclique graph} of $G$, denoted by $KB(G)$, is the intersection graph of the family of all bicliques of $G$.
It was defined and characterized in~\cite{GroshausSzwarcfiterJGT2010}. However, no polynomial time algorithm is known for recognizing biclique graphs.
Refer to~\cite{puppo2,puppo} for some results when restricted to particular graph classes.

In this work we study the following question: Given a biclique graph $G=KB(H)$, is it possible to remove a vertex $q$ of $G$, such that
$G - \{q\}$ is a biclique graph? And if possible, can we obtain a graph $H'$ such that $G - \{q\} = KB(H')$? We show that in general, the answer
of this question is ``no''. However, when the biclique graph has vertices of degree two, we answer ``yes'' to this question. In particular,
our main theorem ensures that if a biclique graph $G=KB(H)$ has a vertex $q$ of degree two, then $G - \{q\}$ is a biclique graph and furthermore, we show how to
obtain a graph $H'$ such that $G - \{q\} = KB(H')$.

This work is motivated by the open computational complexity of the recognition problem (that we suspect to be $NP-$complete), therefore we try to understand more the structure of biclique graphs.
In particular, our result helps us, as well, to know about how a graph $H$ should be, given its biclique graph $G=KB(H)$. 
Moreover, it is also the first result allowing to maintain the property of being a biclique graph after removing a vertex (as far as we know, there are no results of this kind 
even for the more studied subject of clique graphs). Since our result might be generalized (not easily) to higher degrees 
of vertices, one could use these kind of results (along with the proposed conjectures of last section) to restrict the input graph to a more ``controlled'' structure in order to solve 
the recognition problem of biclique graphs. We hope that these tools give some light to this matter. 

Also, as a direct application of our main result, it can help to decide if a graph $G$ is not a biclique graph since we can remove iteratively (if possible) vertices of degree two and if we 
obtain a graph that we know it is not a biclique graph, then $G$ and all intermediate graphs are not biclique graphs. 

This work is organized as follows. In Section $2$ the notation is given. In Section $3$ we discuss the main question of this work and we show why, in general, the answer is ``no'', while in Section $4$ we give a positive answer when biclique graphs have vertices of degree two. In the last section we present some open related problems.

\section{Preliminaries}

Along the paper we restrict to undirected simple graphs. Let $G=(V,E)$ be a graph with vertex set $V(G)$ and edge set $E(G)$, and 
let $n=|V(G)|$ and $m=|E(G)|$. We will use the notation $G' \subseteq G$ whenever $G'$ is an induced subgraph of $G$.
Say that $G$ is a 
\textit{complete graph} when every possible edge belongs to $E$ and say that $G$ is \textit{complete bipartite} when every possible edge between each part of the partition belongs to $E$.
A complete graph of $n$ vertices is denoted $K_{n}$ and a complete bipartite graph on $n$ and $n'$ vertices in each part of the partition respectively, is denoted $K_{n,n'}$.
A \textit{clique} of $G$ is a maximal complete subgraph, while a \textit{biclique} is a maximal induced complete bipartite subgraph 
of $G$. 
We will denote bicliques as $B = B_1 \cup B_2$ where $B_1$ and $B_2$ are the two parts of the biclique.
A set of vertices $I$ is an \textit{independent set} when there is no edge between any pair of them.
The \textit{open neighborhood} of a vertex $v \in V(G)$, denoted by $N(v)$, is the set of vertices adjacent to $v$. 
The \textit{closed neighborhood} of a vertex $v \in V(G)$, denoted by $N[v]$, is the set $N(v) \cup \{v\}$.
The \textit{degree} of a vertex $v$, denoted by $d(v)$, is defined as $d(v) = |N(v)|$. 
Two vertices $v,w \in V(G)$ are called \textit{false-twins} if $N(v)=N(w)$. A graph is \textit{false-twin-free} if it does not contain any false-twins vertices.
Consider all maximal sets of false-twin vertices $Z_1,\ldots,Z_k$ and let $\{z_1,z_2,\ldots,z_k\}$ be a set of representative vertices such that
$z_i \in Z_i$. We call $Tw(G)$ to the graph obtained by the deletion of all vertices of $Z_i - \{z_i\}$, for $i = 1\ldots k$ (definition from~\cite{marinayo}). 
Observe that $Tw(G)$ has no false-twin vertices. A vertex $v\in V(G)$ is \textit{universal} if it is adjacent to all other vertices in $V(G)$.
A \textit{path} of $k$ vertices is denoted by $P_{k}$ and a \textit{cycle} of $k$ vertices is denoted by $C_k$.
We assume that all the graphs of this paper are connected.

A \textit{diamond} is a complete graph with $4$ vertices minus an edge. A \textit{gem} is an induced path with $4$ vertices plus an 
universal vertex.


\section{Removing vertices in biclique graphs}\label{sectionblabla}

We recall the theorem in~\cite{GroshausSzwarcfiterJGT2010} that gives a necessary condition for a graph to be a biclique graph.
The contrapositive of this theorem is useful to prove that a graph is not a biclique graph.

\begin{theorem}[\cite{GroshausSzwarcfiterJGT2010}]\label{tMarina}
Let $G$ be a biclique graph. Then every induced $P_3$ of $G$ is contained in an induced \textit{diamond} or an induced \textit{gem} of $G$ as shown in Figure~\ref{FigtMarina}.
\end{theorem}
\begin{figure}[ht!]
	\centering
	\includegraphics[scale=.4]{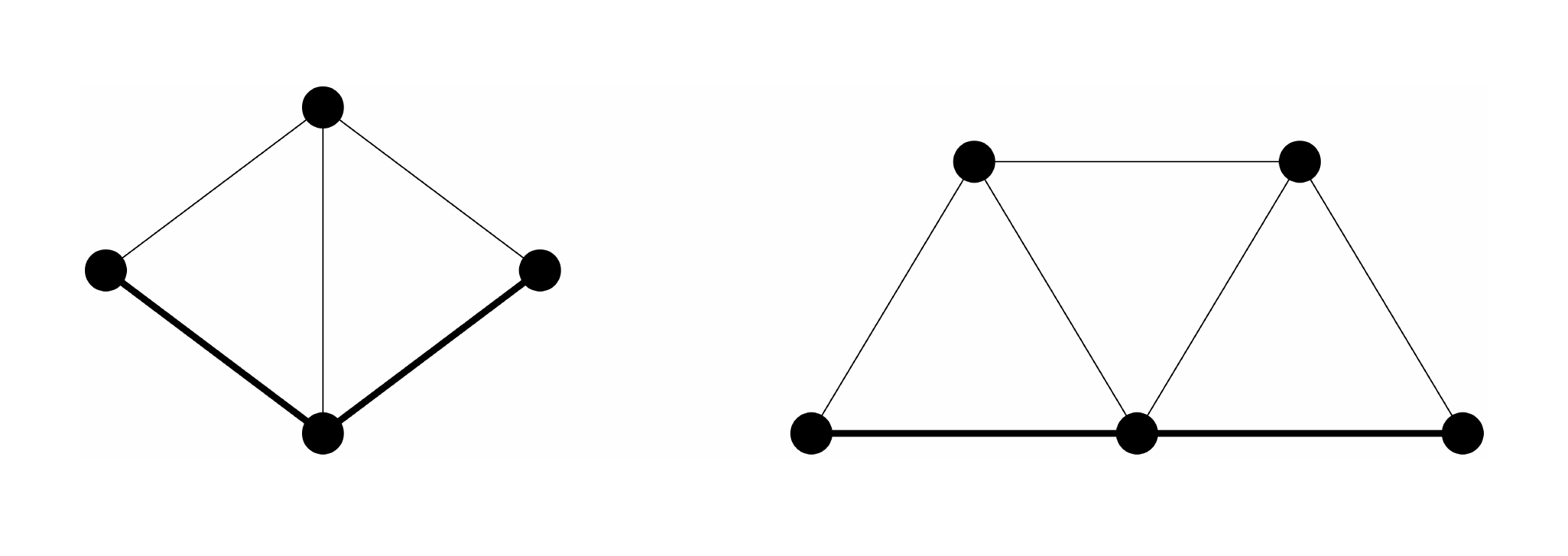}
	\caption{Induced $P_3$ in bold edges contained in a \textit{diamond} and in a \textit{gem} respectively.}
	\label{FigtMarina}
\end{figure}

To answer the question stated in the introduction, consider the  following observation.

\begin{observation}\label{obsLean}
The graph $G=KB(C_k)$, for $k \geq 7$, has no vertex such that, after its removal, the resulting graph is still a biclique graph.
\end{observation}
\begin{proof}
After removing any vertex of $KB(C_k)$, we obtain an induced $P_3$
that is not contained in any induced \textit{diamond} or \textit{gem}, thus for every vertex $q \in KB(C_k)$, $KB(C_k)-\{q\}$ is not a biclique graph by Theorem~\ref{tMarina}. 
\end{proof}

We can conclude that the answer to the question in general is ``no''. See Figure~\ref{c7kb} for an example. 

\begin{figure}[ht]
  \centering
  \includegraphics[scale=.3]{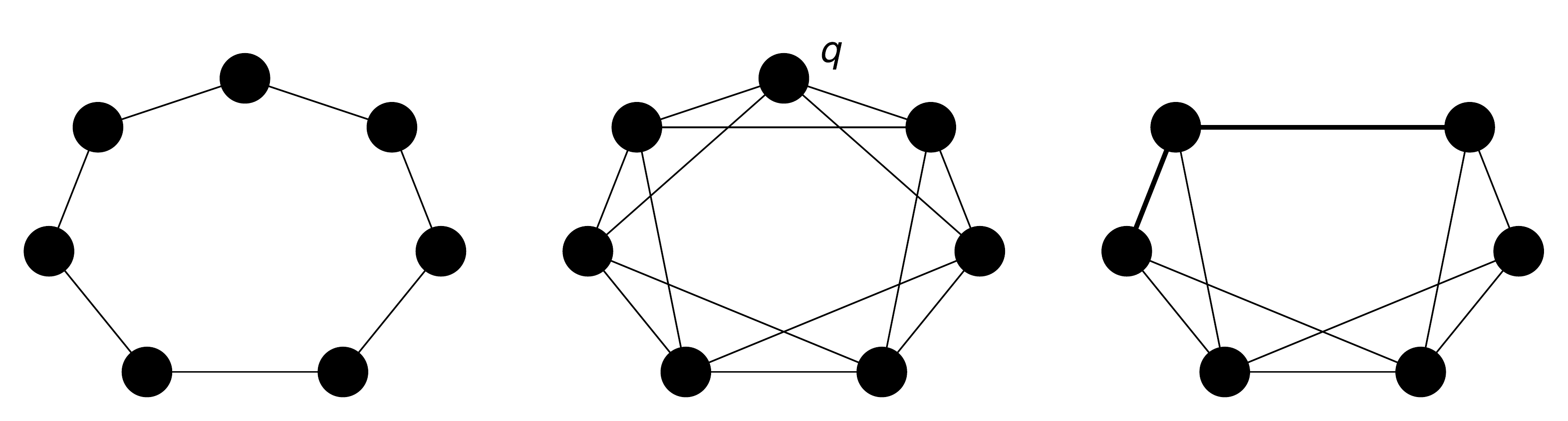}
  \caption{$C_7$, $KB(C_7)$ and $KB(C_7)-\{q\}$. In bold edges we show a $P_3$ not contained in any induced \textit{diamond} or \textit{gem}, therefore $KB(C_7)-\{q\}$ is not a biclique graph.}
  \label{c7kb}
\end{figure}

However, we can answer ``yes'' to the question in a particular case. We will show in the next section that biclique graphs with vertices of degree two, are still biclique graphs after removing them.
Moreover, we show how to obtain a graph $H'$ such that $G - \{q\} = KB(H')$ for a vertex $q$ of degree two.

\section{Biclique graphs with vertices of degree two}\label{sectiongrado2}

The main theorem of this article is the following:

\begin{theorem}\label{teoGrado2}
Let $G=KB(H)$ for some graph $H$ and let $q$ be a vertex of $G$ such that $d(q)=2$.
Then $G-\{q\}$ is a biclique graph. In particular, we can construct a graph $H'$ such that $G-\{q\} = KB(H')$.
\end{theorem}

Note that it is easy to see that the converse implication does not hold. For instance, the crown graph (five vertices
where two are universal and three with degree two) is not a biclique graph \cite{marinaYoArxiv} but after removing any of
the vertices of degree two, we obtain the diamond graph that is a biclique graph.

Before proving Theorem~\ref{teoGrado2} we need several lemmas.

\begin{lemma}\label{lemmaCompu}
Let $H$ be a false-twin-free graph with $n \geq 7$. If there exists a biclique $B$ in $H$
such that either $K_{1,3} \subseteq B$ or $C_4 = B$, then $d(q) \geq 3$ in $G=KB(H)$ where $q$ is the vertex of $G$
corresponding to the biclique $B$ of $H$.
\end{lemma}
\begin{proof}
The proof is by induction on $n = |V(H)|$. For the base case $n=7$, there are only $507$ false-twin-free graphs. We generated them all
using NAUTY C library~\cite{nauty} and verified the lemma for those graphs that satisfy
the hypothesis, i.e., whenever there exists a biclique $B$ such that either $K_{1,3} \subseteq B$ or $C_4 = B$, we
obtained that $d(q) \geq 3$ in $G=KB(H)$ for $q$ the vertex of $G$ corresponding to the biclique $B$ of $H$.
Therefore, the base case holds.

Now, let $H$ be a false-twin-free graph with $n \geq 8$ and let $B$ be a biclique of $H$ such that either $K_{1,3} \subseteq B$ 
or $C_4 = B$. 
Let $A = V(H)-V(B)$ and consider only all vertices $v \in A$ such that $H-\{v\}$ is connected. 
Now, after removing one of these vertices $v$, some false-twin vertices might appear in
$B$ and $A-\{v\}$. Clearly, if for one vertex $v \in A$, $H-\{v\}$ is false-twin-free, then the biclique $B$ still exists. Therefore, as $|V(H - \{v\})| \geq 7$, by induction hypothesis, $d(q) \geq 3$ in 
$G'=KB(H - \{v\})$ and since $H - \{v\}$ is an induced subgraph of $H$ then $G'$ is a subgraph of $G$ (\cite{marinayo}), that is, $d(q) \geq 3$ in $G$ as desired.
We assume now that $H-\{v\}$ has false-twins for all these vertices $v\in A$.

Suppose first that there is a vertex $v \in A$ (among those that $H-\{v\}$ is connected), such that $H-\{v\}$ has only false-twins in the set $A- \{v\}$. Consider the graph $Tw(H-\{v\})$. 
Clearly, it is false-twin-free and the biclique $B$ still exists. Therefore, if $|V(Tw(H-\{v\}))| \geq 7$ the result follows by induction as before, otherwise $|V(Tw(H-\{v\}))| \leq 6$.
In fact, $|V(Tw(H-\{v\}))| = 6$ since $Tw(H-\{v\})$ has no false-twins, $n \geq 8$, $|B| \geq 4$ and we deleted $v$. Now, there are
$61$ false-twin-free graphs with $6$ vertices and it can be checked (using the computer for example) that only three among them have either a $C_4$ or $K_{1,3}$ biclique such that its corresponding vertex in the biclique graph has degree less than three. These graphs are shown in Figure~\ref{fig3graphs}. 
One can verify that wherever we add back the vertex $v$ in order to maintain the property of being 
false-twin-free, we would have at least one more biclique that intersects $B$, that is, the degree of its corresponding vertex in the biclique graph is at least three as we wanted to prove.

\begin{figure}[ht]
  \centering
  \includegraphics[scale=.3]{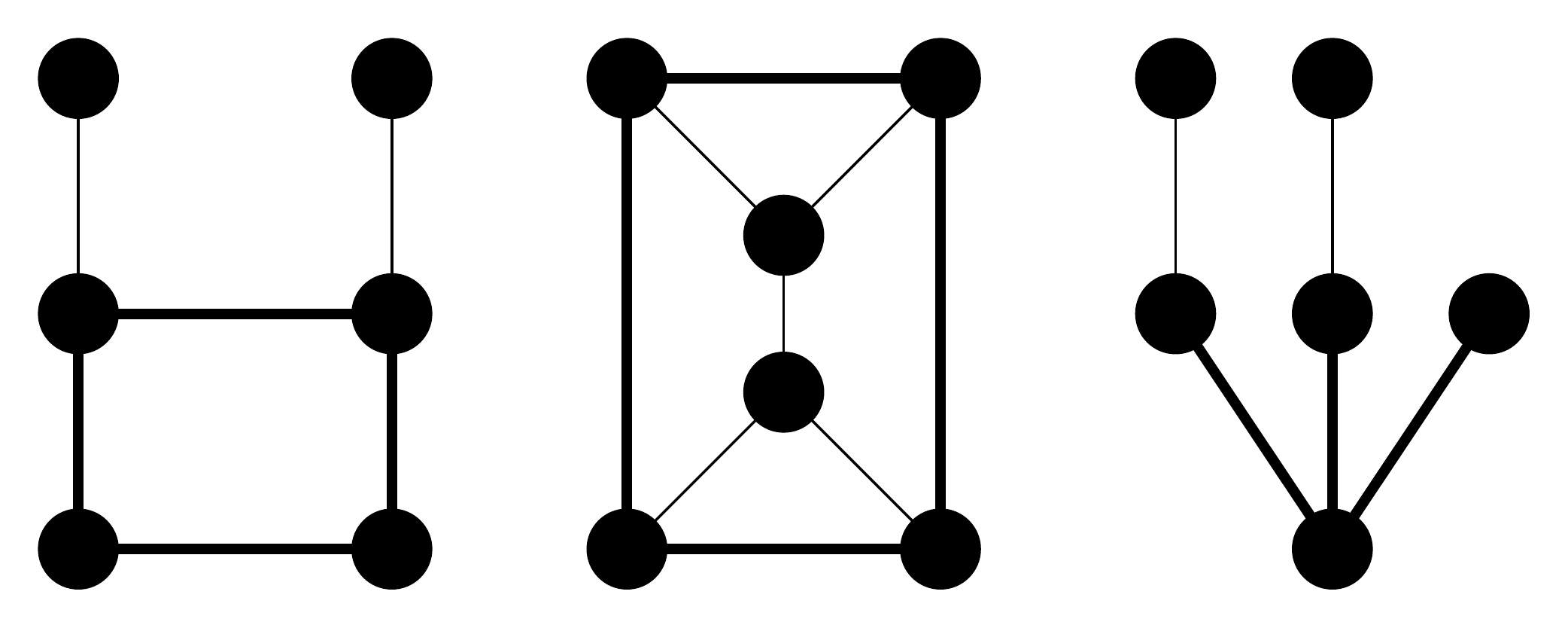}
  \caption{Unique three false-twin-free graphs with $6$ vertices and a biclique isomorphic to either $C_4$ or $K_{1,3}$ such that in their biclique graph, the vertex corresponding to
  that biclique has degree two. In bold edges these bicliques are shown.}
  \label{fig3graphs}
\end{figure} 
 
Suppose next that $H-\{v\}$ has false-twins in the set $B$. Moreover, we can assume that for all $v \in A$ 
(again, among those that $H-\{v\}$ is connected), $H-\{v\}$ has false-twin vertices in $B$, since if $H-\{v\}$ has 
false-twins only in $A-\{v\}$, we proceed exactly as in the previous case.
 
We show first that $|A| \geq 3$. 
Clearly, as $n \geq 8$, if $C_4=B$, then $|A| \geq 4$ and there is nothing to show. Now, if $K_{1,3} \subseteq B$ and 
$|A| \leq 2$, again, since $n \geq 8$, we have that either $K_{1,s} \subseteq B$ for $s \geq 5$, $K_{2,4} = B$ or $K_{3,3} = B$. In the first case, $H$ has false-twin vertices, since only two vertices in $A$ are not enough to make $s \geq 5$ vertices of $K_{1,s} \subseteq B$ not false-twins, therefore a contradiction. For the other two cases we have that $|A|=2$ thus $n=8$ as $H$ has no 
false-twins\footnote{We remark that there exist $7442$ false-twin-free graphs on $8$ vertices (\cite{nauty}) therefore this case could also be easily verified with the computer.}. 
If $K_{2,4} = B$, then the only way to make the $4$ vertices in one part of the partition of $K_{2,4}$
not false-twins using two vertices of $A$ is shown in Figure~\ref{lemacompu1}. Now depending on the adjacencies
between vertices $v,w$ and $x,y$ (at least one edge must exist as $x,y$ are not false-twins), we can easily see that we always obtain at least $3$ other bicliques that intersect $B$, 
that is, $d(q) \geq 3$ in $G=KB(H)$. Finally, if $K_{3,3} = B$, let $v,w \in A$. Since $H$ is false-twin-free, each of $v$ and $w$ must be adjacent to at least one and at most two vertices of each part of the partition
of $B$, and moreover, their neighbors of each part of the partition cannot be the same. Therefore each of $v$ and $w$ is not adjacent to at least one and at most two vertices of each part of the partition of $B$ and these
sets of non-neighbors of each part of the partition cannot be the same neither. Now, taking $v$ along with the maximal set of adjacent vertices to $v$ of one part of the partition and the maximal set of non-adjacent vertices
to $v$ on the other part, we obtain one biclique. Doing the same but taking adjacent and non-adjacent vertices to $v$ of opposite parts of the partition we obtain a second biclique. Finally, using the same argument for $w$, we obtain two more bicliques, that is, we have in total at least $4$ bicliques that intersect $B$, so $d(q) \geq 4$ in $G=KB(H)$. 
Thus we assume in what follows that $|A| \geq 3$.

\begin{figure}[ht]
  \centering
  \includegraphics[scale=.4]{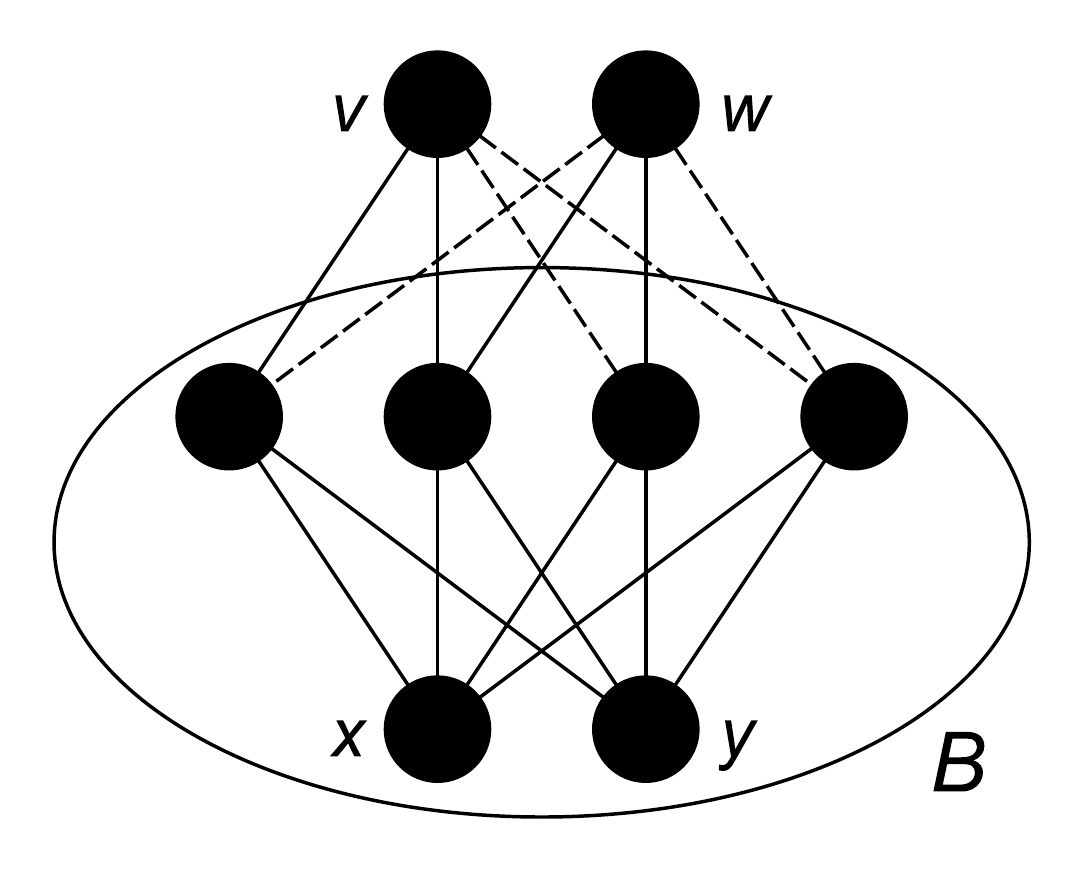}
  \caption{Graph on $8$ vertices with a biclique $B=K_{2,4}$ and no false-twins. Dotted edges are not present however edges
  between $v,w$ and $x,y$ might exist (at least one must, as $x,y$ are not false-twins).}
  \label{lemacompu1}
\end{figure}

We show now that if $v,w \in A$, then they belong to two different bicliques that intersect $B$. This would imply, as $|A| \geq 3$,
that we have $3$ other bicliques in $H$ that intersect $B$, that is, $d(q) \geq 3$ in $G$. Recall now that $H-\{v\}$ and
$H-\{w\}$ have false-twins in $B$. Let $x_1,x_2$ be the false-twins of $B$ in $H-\{v\}$, with $v$ 
adjacent to $x_1$ and let $y_1,y_2$ be the false-twins of $B$ in $H-\{w\}$, with $w$ adjacent to $y_1$.
We have the following cases:
\begin{itemize}
 \item False-twins belong to different parts of the partition of $B$ (see Figure~\ref{lemacompu2}). 
 Now $\{v\}\cup\{x_1\}$ and $\{w\}\cup\{y_1\}$ belong to two different bicliques that intersect $B$ unless $v$ is adjacent to $w$,
 $x_1$ not adjacent to $w$ and $y_1$ not adjacent to $v$. In this case $\{v\}\cup\{x_1\}$ and $\{w\}\cup\{y_1\}$ would belong to the same biclique containing $\{w,x_1\}\cup\{v,y_1\}$. However, we would also obtain that $\{x_1\}\cup\{v,y_1,y_2\}$ and $\{y_1\}\cup\{w,x_1,x_2\}$ are contained in two different  bicliques that intersect $B$ as $y_2,v$ and $x_2,w$ are not adjacent ($N(x_1)-\{v\}=N(x_2)$ and $N(y_1)-\{w\}=N(y_2)$).
 \begin{figure}[ht]
  \centering
  \includegraphics[scale=.4]{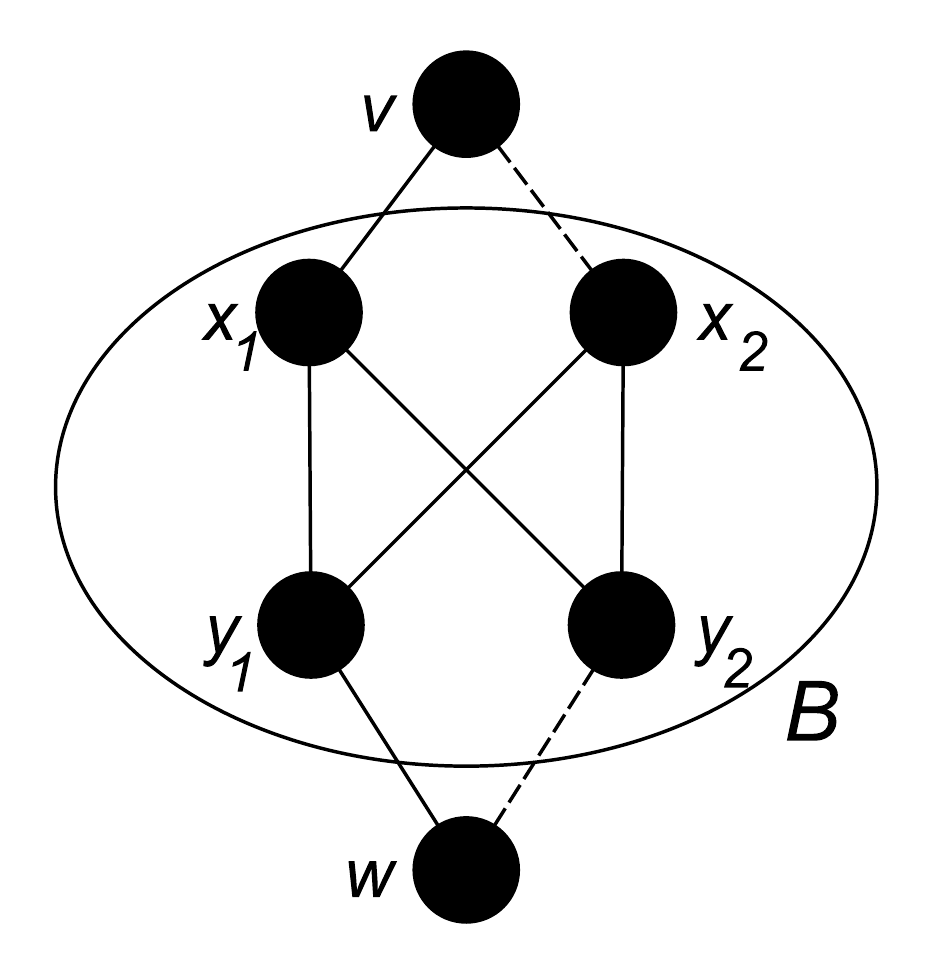}
  \caption{False-twins belong to different parts of the partition of biclique $B$.}
  \label{lemacompu2}
\end{figure}
 \item False-twins belong to the same part of the partition of $B$. We have more cases now depending on $x_1,x_2,y_1$ and $y_2$.
 \begin{itemize}
  \item $x_1,x_2,y_1,y_2$ are all different (see Figure~\ref{lemacompu3}). 
  We have that $\{v\}\cup\{x_1\}$ and $\{w\}\cup\{y_1\}$ belong to two different bicliques that intersect $B$ unless $v$ is not adjacent to $w$, $x_1$ adjacent to $w$ and $y_1$ adjacent to $v$, in which case we obtain that $\{w\}\cup\{x_1,x_2,y_1\}$ and $\{v\}\cup\{y_1,y_2,x_1\}$ are contained in two different bicliques intersecting $B$ as $y_2,v$ and $x_2,w$ are adjacent ($N(x_1)-\{v\}=N(x_2)$ and $N(y_1)-\{w\}=N(y_2)$).
\begin{figure}[ht]
  \centering
  \includegraphics[scale=.4]{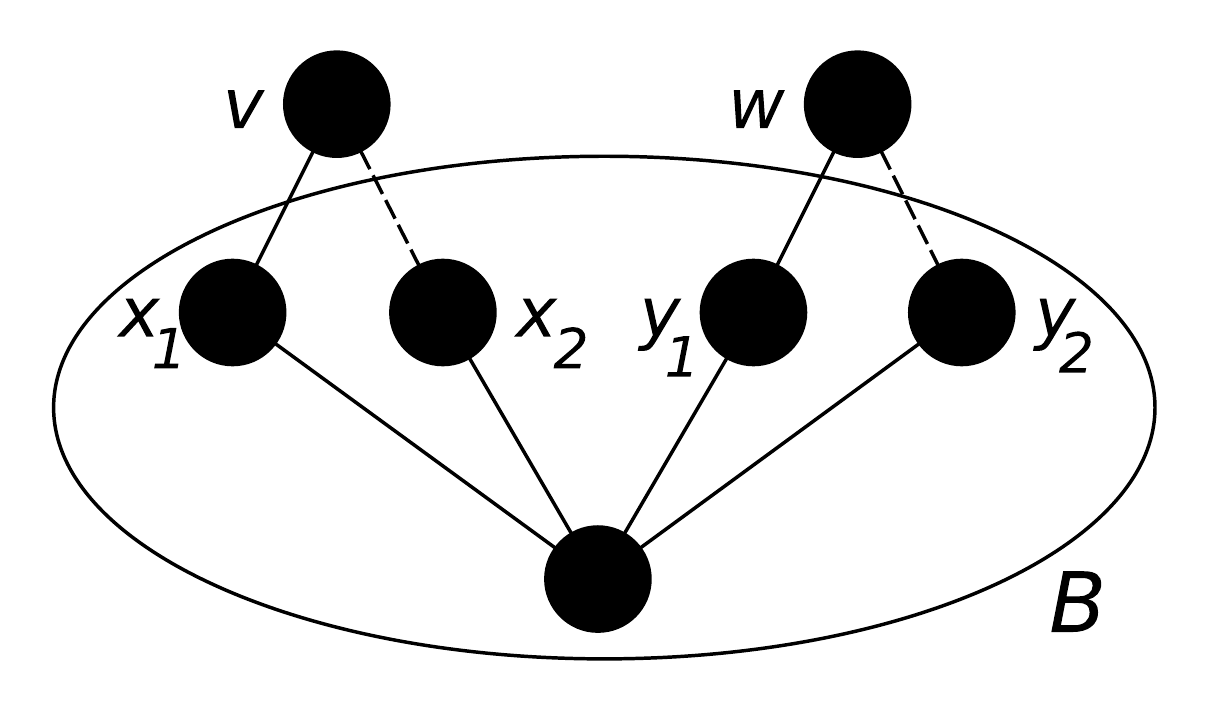}
  \caption{Case where $x_1,x_2,y_1,y_2$ are all different.}
  \label{lemacompu3}
\end{figure}
  \item $x_1 = y_1$ and $x_2=y_2$ (or $x_1 = y_2$ and $x_2=y_1$). Clearly this case cannot happen.
  \item $x_1 \neq y_1$ and $x_2 = y_2$. Since $v,y_1$ and $w,x_1$ are not adjacent (otherwise $v,x_2$ and $w,x_2$ would be
  adjacent as well which is a contradiction), we obtain that $\{v\}\cup\{x_1\}$ and $\{w\}\cup\{y_1\}$ belong to two different bicliques that intersect $B$.
  \item $x_1 \neq y_2$ and $x_2 = y_1$. Observe that $w,x_1$ should be adjacent.
  In this case, $\{v\}\cup\{x_1\}$ and $\{w\}\cup\{x_1,x_2\}$ belong to two different bicliques that intersect $B$.
 \end{itemize}
\end{itemize}
Since we covered all cases, we conclude that $d(q) \geq 3$ in $G = KB(H)$ for $q$ the vertex of $G$ corresponding to the biclique $B$ of $H$, as desired.
\end{proof}

\begin{lemma}\label{lemaLoco}
 Let $G=KB(H)$ for some graph $H$ without false-twin vertices and $|V(H)| \geq 7$. Let $q$ be a vertex of $G$ such that $d(q) = 2$
 and let $B$ be its corresponding biclique in $H$. Then we have one of the following:
 \begin{enumerate}
  \item $B=K_{1,1}$ and $H$ belongs to the family of the following graphs where $I$ is an independent set. 
  \begin{center}
   \includegraphics[scale=.3]{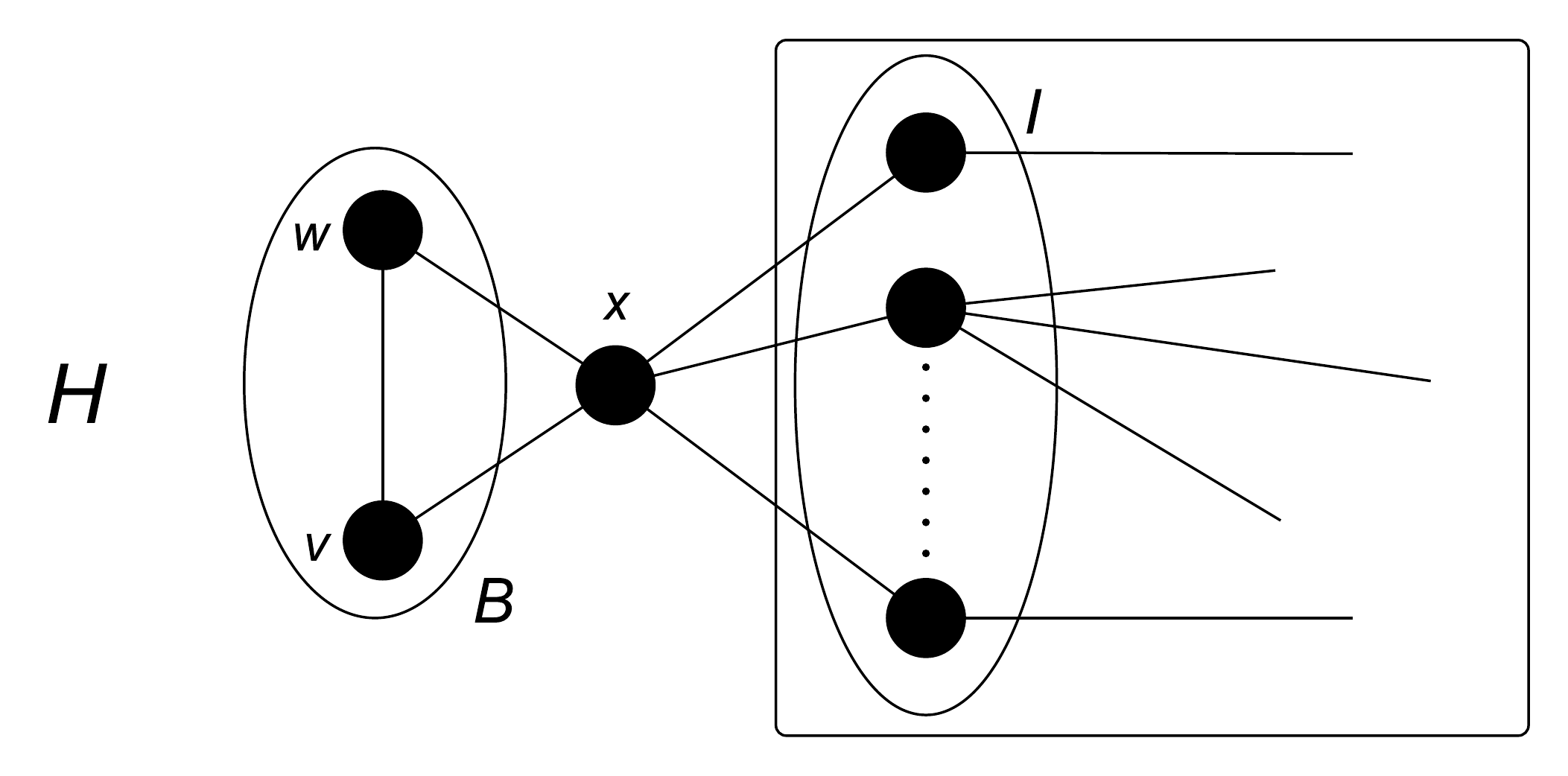}
  \end{center}
  \item $B=K_{1,2}$ and $H$ belongs to the family of the following graphs where $N(a)=\{b\}$ and $N(b)=\{a,c\}$.
	\begin{center}
   \includegraphics[scale=.3]{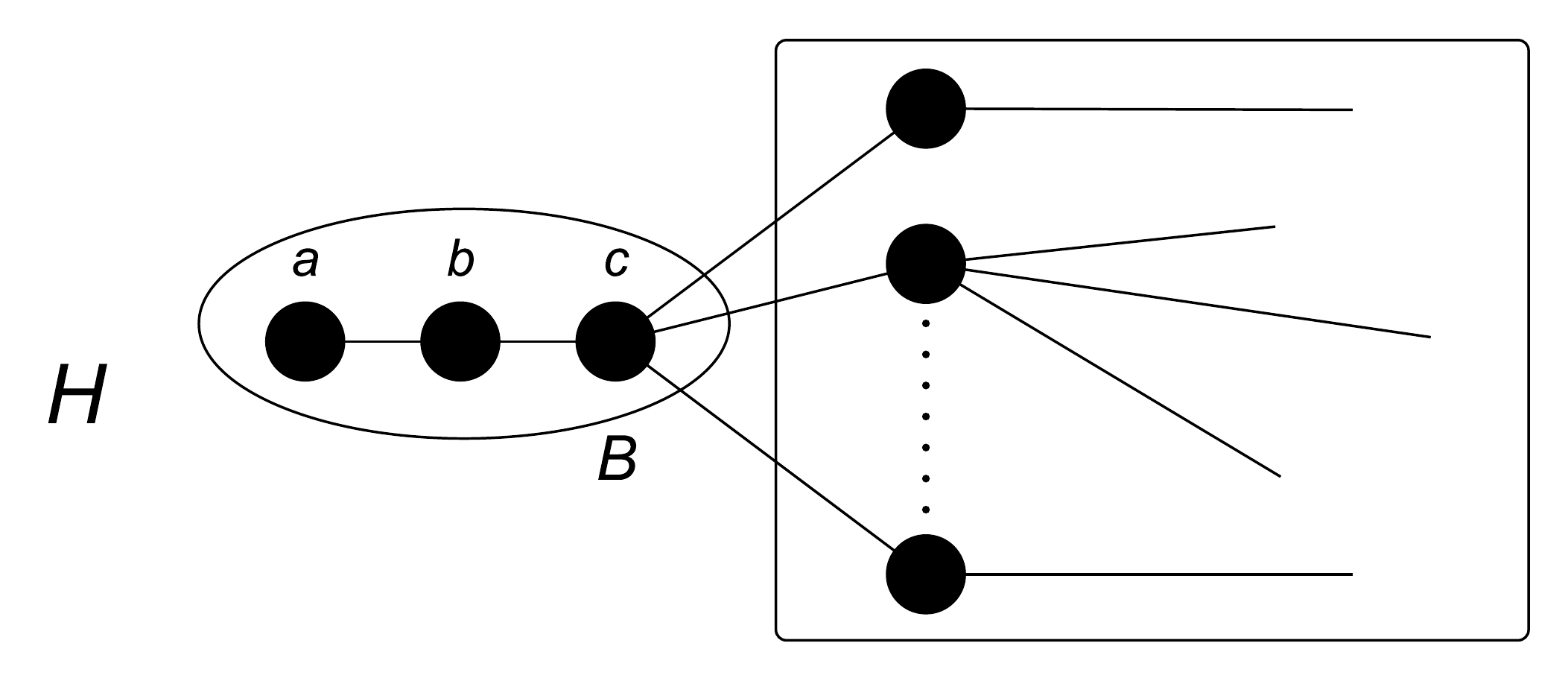}
  \end{center}
 \end{enumerate}
\end{lemma}
\begin{proof}
 Suppose first that $B=K_{1,1}=\{v\}\cup\{w\}$. Since the biclique $B$ is a $K_{1,1}$, we must have that $N[v]=N[w]$. 
 By contradiction, suppose that
 $v,w$ have more than one common neighbor. Let $x_1,x_2$ be two neighbors of $v$ and $w$. If $x_1$ and $x_2$ are adjacent
 then we have that $\{v\}\cup\{x_1\}$, $\{v\}\cup\{x_2\}$, $\{w\}\cup\{x_1\}$ and $\{w\}\cup\{x_2\}$ are contained in
 four different bicliques that intersect $B$. This is a contradiction since $d(q)=2$ in $G$. Now if $x_1$ and $x_2$ 
 are not adjacent, since $H$ has no false-twins, there exists a vertex $x_1'$ adjacent to $x_1$ and not adjacent to $x_2$, or
 a vertex $x_2'$ adjacent to $x_2$ and not adjacent to $x_1$. Assume without loss of generality the first case.
 Now we have that $\{v\}\cup\{x_1,x_2\}$, $\{w\}\cup\{x_1,x_2\}$, $\{x_1\}\cup\{v,x_1'\}$ and $\{x_1\}\cup\{w,x_1'\}$ are contained in four different bicliques that intersect $B$, again a contradiction. Therefore $v$ and $w$ have exactly one common neighbor, say $x$, that is, $N[v]=N[w] = \{v,w,x\}$. Now as $|V(H)| \geq 7$,
 there are more vertices in the graph. Suppose that $x$ has two adjacent vertices, say $x_1,x_2$. We obtain that
 $\{x\}\cup\{v,x_1\}$, $\{x\}\cup\{v,x_2\}$, $\{x\}\cup\{w,x_1\}$ and $\{x\}\cup\{w,x_2\}$ are contained in four different bicliques intersecting $B$, a contradiction. We conclude then that 
$N(x) - \{v,w\}$ in an independent set, that is, \textit{(1)} holds.
 
 Suppose last that $B$ is bigger than $K_{1,1}$. Now if $K_{1,3} \subseteq B$ or $C_4 = B$, by Lemma~\ref{lemmaCompu}, we
 have that $d(q) \geq 3$ in $G$, that is, a contradiction. Therefore $B=K_{1,2}$. Let $B = \{b\} \cup \{a,c\}$ as shown in the statement of the lemma.
 We will show that $N(a)=\{b\}$ and $N(b)=\{a,c\}$ by contradiction.
 
 \textbf{Case (A)} Suppose first that $b$ has another neighbor $b'$. Now, as $B=K_{1,2}$, we have the following cases depending on the adjacencies of $b'$.

 \textbf{Case (A1)} $b'$ is adjacent to $a$ and $c$: as $H$ has no false-twins, there exists a vertex $a'$ adjacent to $a$ and not adjacent to $c$ (equivalent for a vertex $c'$ adjacent to $c$
 and not to $a$). Therefore we have that the sets $\{b\} \cup \{b'\}$ and $\{b'\} \cup \{a,c\}$ are contained in two different bicliques that intersect $B$. Finally,
 if $a'$ is not adjacent to $b$ then $\{a\} \cup \{a',b\}$ is contained in a third biclique that intersects $B$, otherwise, $\{b\} \cup \{a',c\}$
 is contained in a third biclique that intersects $B$ (see Figure~\ref{caseA1}). In both cases we obtain a contradiction, therefore $b'$ cannot be adjacent to both $a$ and $c$.
 
\begin{figure}[ht]
  \centering
  \includegraphics[scale=.4]{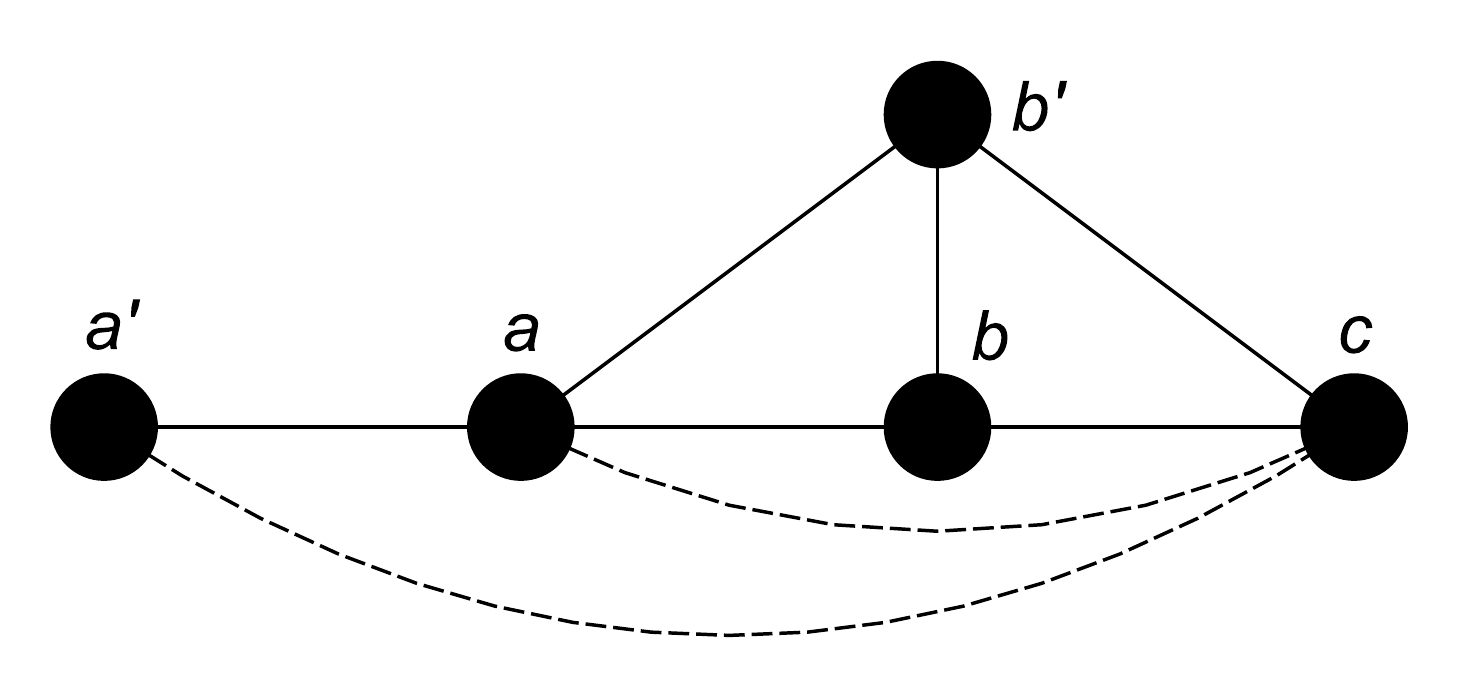}
  \caption{\textbf{Case (A1)}.}
  \label{caseA1}
\end{figure} 
 
 \textbf{Case (A2)} $b'$ is adjacent to $a$ and not adjacent to $c$ (symmetric for $b'$ adjacent to $c$ and not adjacent to $a$): 
 We have more cases here.

 \textbf{Case (A2a)} Suppose now that $b$ has another neighbor $b''$. As $B=K_{1,2}$, we have the following cases:

 	\textbf{Case (A2a1)} $b''$ is adjacent to $c$ and not adjacent to $a$: Here we have that $\{b\} \cup \{b',c\}$, $\{b\} \cup \{b'',a\}$ and $\{a\} \cup \{b'\}$ are contained in 
 	three different bicliques respectively, each of them intersecting $B$, that is, a contradiction.
 	
 	\textbf{Case (A2a2)} $b''$ is adjacent to $a$ and not adjacent to $c$: Observe first that if $b'$ is adjacent to $b''$, we obtain that $\{a,b,b',b''\}$ induces a $K_4$, that is, 
 	there are four different bicliques that intersect $B$ which is a contradiction. Then $b'$ is not adjacent to $b''$ and since $H$ has no false-twins, there exists a vertex $\tilde{b}$ adjacent 
 	to $b'$ and not adjacent to $b''$ (see Figure~\ref{caseA2a2}). Now $\{b\} \cup \{b',b'',c\}$ and $\{a\} \cup \{b',b''\}$ are contained in two different bicliques respectively, 
 	intersecting $B$. Finally, if $a$ is not adjacent to $\tilde{b}$, then $\{b'\} \cup \{a,\tilde{b}\}$ is contained in a third biclique intersecting $B$, otherwise 
 	$\{a\} \cup \{b'',\tilde{b}\}$ is contained in a third biclique intersecting $B$. In both cases we obtain a contradiction. 	

\begin{figure}[ht]
  \centering
  \includegraphics[scale=.4]{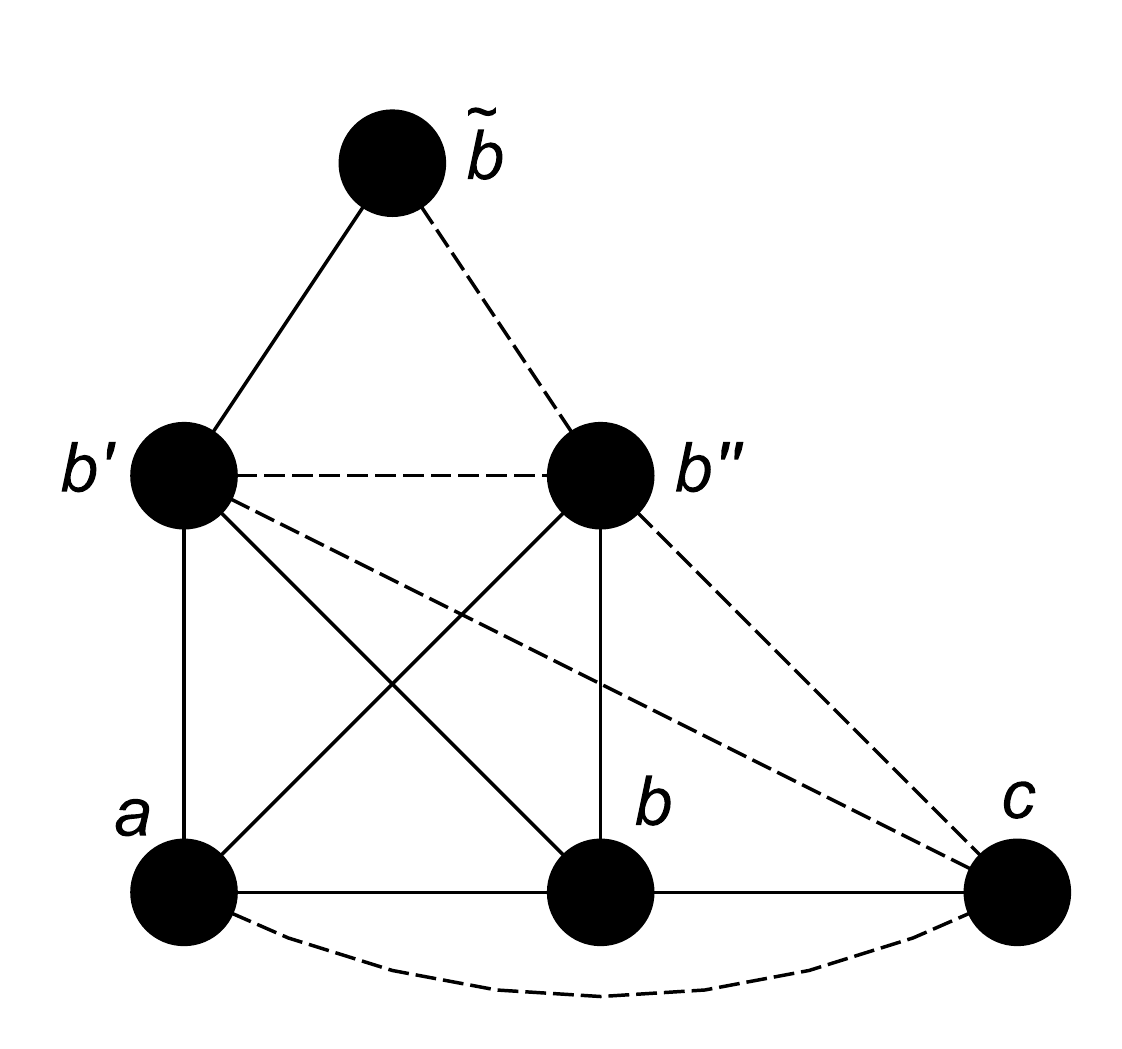}
  \caption{\textbf{Case (A2a2)}.}
  \label{caseA2a2}
\end{figure}

 We can conclude that $b$ does not have other neighbors than $a,b'$ and $c$. 
 
 \textbf{Case (A2b)} Suppose now that $a$ has another neighbor $a'$. Clearly, $a'$ is not adjacent to $b$ (\textbf{Case (A2a)}) and
 $a'$ is not adjacent to $c$, as $B=K_{1,2}$. In this case we obtain that the sets $\{b\} \cup \{b',c\}$, $\{a\} \cup \{b'\}$ and $\{a\} \cup \{a',b\}$ are contained in three different bicliques
 respectively, intersecting $B$, a contradiction. We conclude then that $a$ does not have other neighbors than $b$ and $b'$.
 
 \textbf{Case (A2c)} Suppose now that $b'$ has another neighbor $b''$. As shown before, $b''$ cannot be adjacent to either $a$ or $b$. Now, if $b''$ is not adjacent to $c$, we obtain that
 $\{b\} \cup \{b',c\}$, $\{b'\} \cup \{b'',a\}$ and $\{b'\} \cup \{b'',b\}$ are contained in three different bicliques respectively, intersecting $B$, which is a contradiction.
 Therefore, $b''$ is adjacent to $c$. So far we have that $\{b'\} \cup \{b'',a\}$ and $\{b',c\} \cup \{b'',b\}$ are contained in two different bicliques intersecting $B$ (see Figure~\ref{caseA2c}).
 Now, as $|V(H)| \geq 7$, the graph contains more vertices: 
 
\begin{figure}[ht]
  \centering
  \includegraphics[scale=.4,trim={0 0 0 4cm},clip]{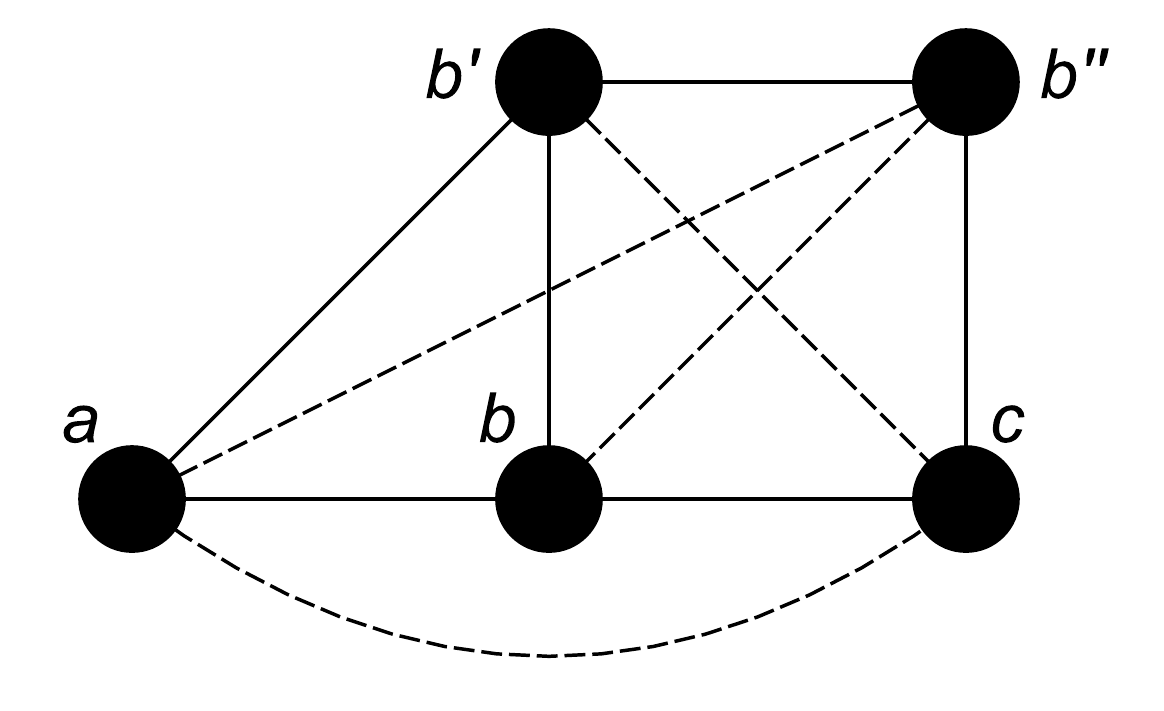}
  \caption{\textbf{Case (A2c)}.}
  \label{caseA2c}
\end{figure}
 
 \textbf{Case (A2c1)} Suppose $c$ has another neighbor $c'$. Clearly, $c'$ is not adjacent to $a$ and $b$. If $c'$ is not adjacent to $b'$,
 we obtain that $\{b\} \cup \{b',c\}$, $\{a\} \cup \{b'\}$ and $\{c\} \cup \{c',b\}$ are contained in three different bicliques respectively, intersecting $B$, a contradiction.
 Therefore, $c'$ is adjacent to $b'$. We have the following cases now:

 \textbf{Case (A2c1a)} $c'$ is adjacent to $b''$: In this case we have that $\{b'\} \cup \{b'',a\}$, $\{b'\} \cup \{c',a\}$, $\{b',c\} \cup \{b,c'\}$ and $\{b',c\} \cup \{b'',b\}$ 
 are contained in four different bicliques respectively, intersecting $B$, which is a contradiction.
 
 \textbf{Case (A2c1b)} $c'$ is not adjacent to $b''$: Now, as $H$ has no false-twin vertices, there exists a vertex $\tilde{c}$ adjacent to $c'$ and not adjacent to $b''$ (equivalent for 
 a $\tilde{b}$ adjacent to $b''$ and not adjacent to $c'$). We have that $\{b'\} \cup \{a,b'',c'\}$ and $\{b',c\} \cup \{b'',b,c'\}$ are contained in two different bicliques respectively, 
  intersecting $B$ (see Fig~\ref{caseA2c1b}). Finally, if $\tilde{c}$ is not adjacent to $c$, $\{c'\} \cup \{c,\tilde{c}\}$, is contained in a third biclique that intersects $B$, 
  otherwise $\{c\} \cup \{\tilde{c},b\}$ is contained in a third biclique. In both cases, that is a contradiction.
  
 \begin{figure}[ht]
  \centering
  \includegraphics[scale=.4,trim={0 0 0 3cm},clip]{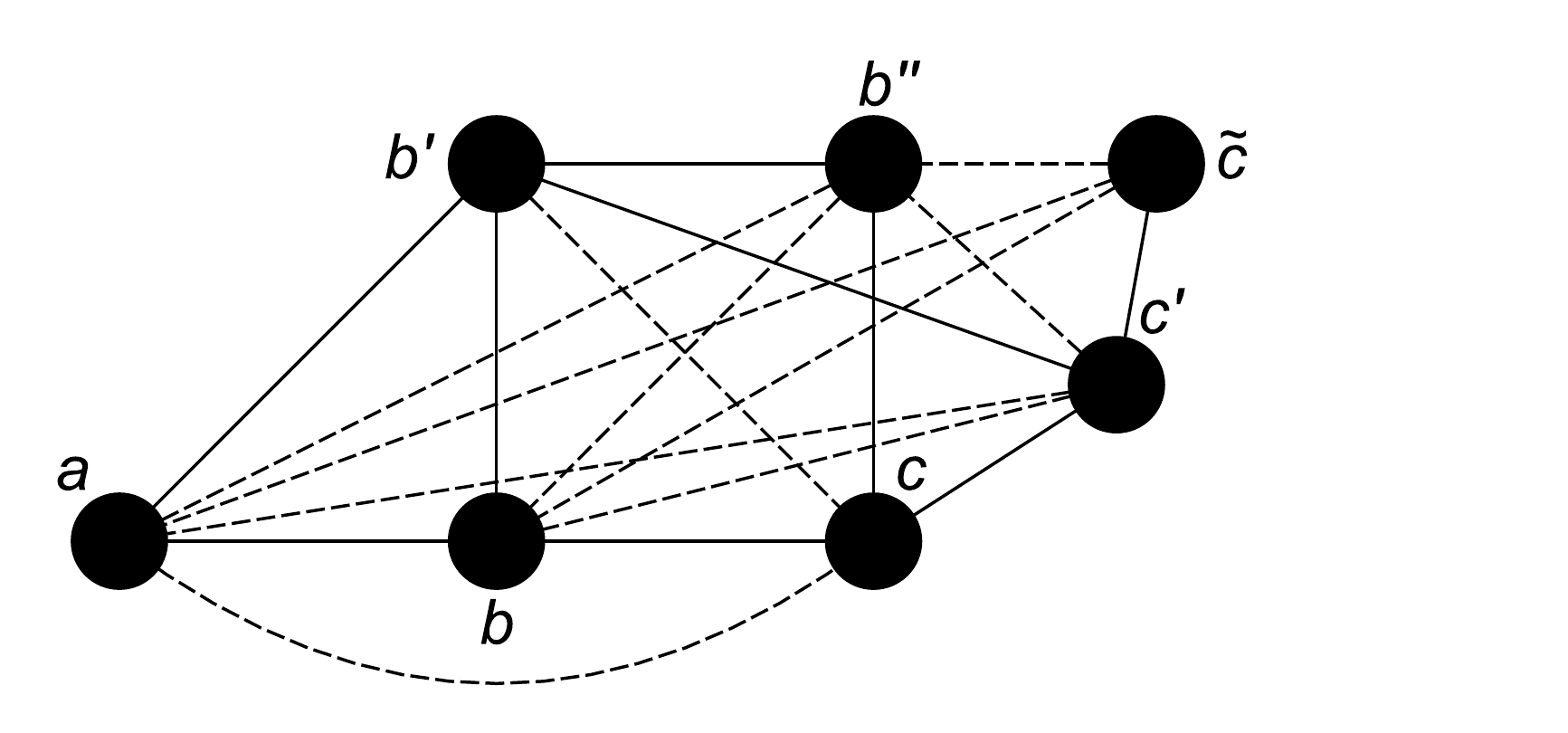}
  \caption{\textbf{Case (A2c1b)}.}
  \label{caseA2c1b}
\end{figure}  
 
 We conclude then that $c$ does not have other neighbors than $b$ and $b''$.
 
 \textbf{Case (A2c2)} Suppose $b'$ has another neighbor $\tilde{b}$. If $\tilde{b}$ is not adjacent to $c$, refer to \textbf{Case (A2c)} (when $b''$ is not adjacent to $c$). Otherwise, 
 refer to \textbf{Case (A2c1)} (when $c'$ is adjacent to $b'$).
 We can conclude that $b'$ has no other neighbors than $a,b$ and $c'$. 
 
 \textbf{Case (A2c3)} Suppose $b''$ has another neighbor $\tilde{b}$. Observe that $\tilde{b}$ cannot be adjacent to any of $a$ (\textbf{Case (A2b)}), $b$ (\textbf{Case (A2a)}), $c$ (\textbf{Case (A2c1)}) and $b'$ (\textbf{Case (A2c2)}).
 Therefore $\{a\} \cup \{b'\}$, $\{b,b''\} \cup \{b',c\}$ and $\{b''\} \cup \{c,b',\tilde{b}\}$ are contained in three different bicliques respectively, each of them intersecting $B$, 
 that is, a contradiction. Thus, $b''$ cannot have other neighbors than $b'$ and $c$.

 As in these three cases we obtained a contradiction, we can conclude that $b'$ has no other neighbors than $a$ and $b$.

 \textbf{Case (A2d)} Suppose now that $c$ has another neighbor $c'$. If $c'$ is adjacent to $b'$, then refer to \textbf{Case (A2c)} (when $b''$ is adjacent to $c$), otherwise we have that 
 $\{a\} \cup \{b'\}$, $\{b\} \cup \{b',c\}$ and $\{c\} \cup \{c',b\}$ are contained in three different bicliques respectively, that intersect $B$, which is a contradiction.
 Thus, $c$ cannot have other neighbors than $b$. 
 
 To conclude \textbf{Case (A2)}, if $b'$ is adjacent to $a$ and not adjacent to $c$, we obtain a contradiction.

 To summarize \textbf{Case (A)}, whenever we supposed that $b$ had another neighbor $b'$, we have obtained a contradiction. So, we conclude that $b$ has no other neighbors than $a$ and $c$,
 that is, $N(b)=\{a,c\}$. 
 
 \textbf{Case (B)} Suppose finally that $a$ has a neighbor $a'$ and $c$ has a neighbor $c'$.
 Clearly $a' \neq c'$ as $B = K_{1,2}$ and $b$ has only $a$ and $c$ as neighbors (\textbf{Case (A)}). Using the same argument, we obtain that $a'$ is not adjacent to $c$
 neither $c'$ to $a$. Observe now that if $a'$ is adjacent to $c'$, we obtain that $\{a,b,c,c',a'\}$ induces a $C_5$, that is, there are four different bicliques
 that intersect $B$ which is a contradiction. 
 Suppose now that $a$ has another neighbor $a''$ (equivalent if $c$ has another neighbor). Like before, $a''$ is not adjacent to $c$ and $c'$.
 Now if $a'$ is adjacent to $a''$, we obtain that $\{a\} \cup \{a',b\}$, $\{a\} \cup \{a'',b\}$ and $\{c\} \cup \{c',b\}$ are contained in three different bicliques
 respectively, each of them intersecting $B$, that is, a contradiction. Therefore, $a'$ is not adjacent to $a''$. Now as $H$ is false-twin-free, there
 exists a vertex $\tilde{a}$ adjacent to $a'$ and not adjacent to $a''$. Clearly $\tilde{a} \neq c'$, otherwise we have an induced $C_5$ so four different bicliques 
 intersecting $B$. Therefore,  we obtain that $\{a\} \cup \{a',a'',b\}$, $\{a'\} \cup \{a,\tilde{a}\}$ and $\{c\} \cup \{c',b\}$ are contained in three different bicliques
 respectively, intersecting $B$, a contradiction.
 We can conclude that either $a$ or $c$ has no other neighbors than $b$. If we suppose that $a$ does not have other neighbors, we obtain that
 $N(a)=\{b\}$.
  
To finish the proof, combining \textbf{Cases (A)} and \textbf{(B)}, we obtain that $N(a)=\{b\}$ and $N(b)=\{a,c\}$ as desired, that is, \textit{(2)} holds.
\end{proof}

Now we can give the proof of Theorem~\ref{teoGrado2}.
\vspace*{4mm}

\prooff{\ref{teoGrado2}} Let $G=KB(H)$ and let $q$ be a vertex of $G$ such that $d(q)=2$. We want to show that
$G-\{q\}$ is a biclique graph and in particular, that we can construct a graph $H'$ such that $G- \{q\} = KB(H')$.
We can assume that $H$ is false-twin-free as $G = KB(H) = KB(Tw(H))$ (\cite{marinayo}). We can also assume that $|V(H)| \geq 7$ since
it is easy to check the theorem if $H$ is false-twin-free and $|V(H)| \leq 6$ (there are only $61$ graphs with $6$ vertices, $11$ with $5$ and $3$ with $4$ as we generated them with NAUTY C library~\cite{nauty}. See Appendix 
of~\cite{marinaYoArxiv} for a list of all biclique graphs up to $6$ vertices). 

Now by Lemma~\ref{lemaLoco} we have that if $B$ is the biclique in $H$ corresponding to the vertex $q$ of $G$, then either $B=K_{1,1}$ and $H$ belongs to the first family of graphs shown in the statement of 
Lemma~\ref{lemaLoco} where $I$ is an independent set, or $B=K_{1,2}$ and $H$ belongs to the second family of graphs shown in the statement of Lemma~\ref{lemaLoco} where $N(a)=\{b\}$ and $N(b)=\{a,c\}$.

In the first case observe that we have the biclique $B$ that intersects only the following two bicliques: $B_1 = \{x\} \cup (N(x) - \{w\})$ and $B_2=\{x\} \cup (N(x) - \{v\})$.
Consider the graph $H'$ obtained as follows: take $H$, remove $v$ and $w$, add a copy of vertex $x$ (with the same neighborhood), say $y$, and finally add a vertex $x'$ adjacent to $x$ and a vertex 
$y'$ adjacent to $y$ (see Figure~\ref{lemalocografo}). Note that $H'$ has exactly one biclique less than $H$ (biclique $B$), moreover, we can associate bicliques $B_1,B_2$ of $H$ to the bicliques 
$B_1' = \{x\} \cup N(x), B_2'= \{y\} \cup N(y)$ of $H'$, respectively. Clearly, as $N(x) - \{x'\} = N(y) - \{y'\}$, $B_1'$ and $B_2'$ intersect, and since this set is equal to $N(x) - \{v,w\}$ 
of $H$, $B_1'$ and $B_2'$ have the same intersections in $H'$ than $B_1$ and $B_2$ of $H$ (minus the biclique $B$). We can conclude then, that $KB(H') = G- \{q\}$.  

 \begin{figure}[ht]
  \centering
  \includegraphics[scale=.4]{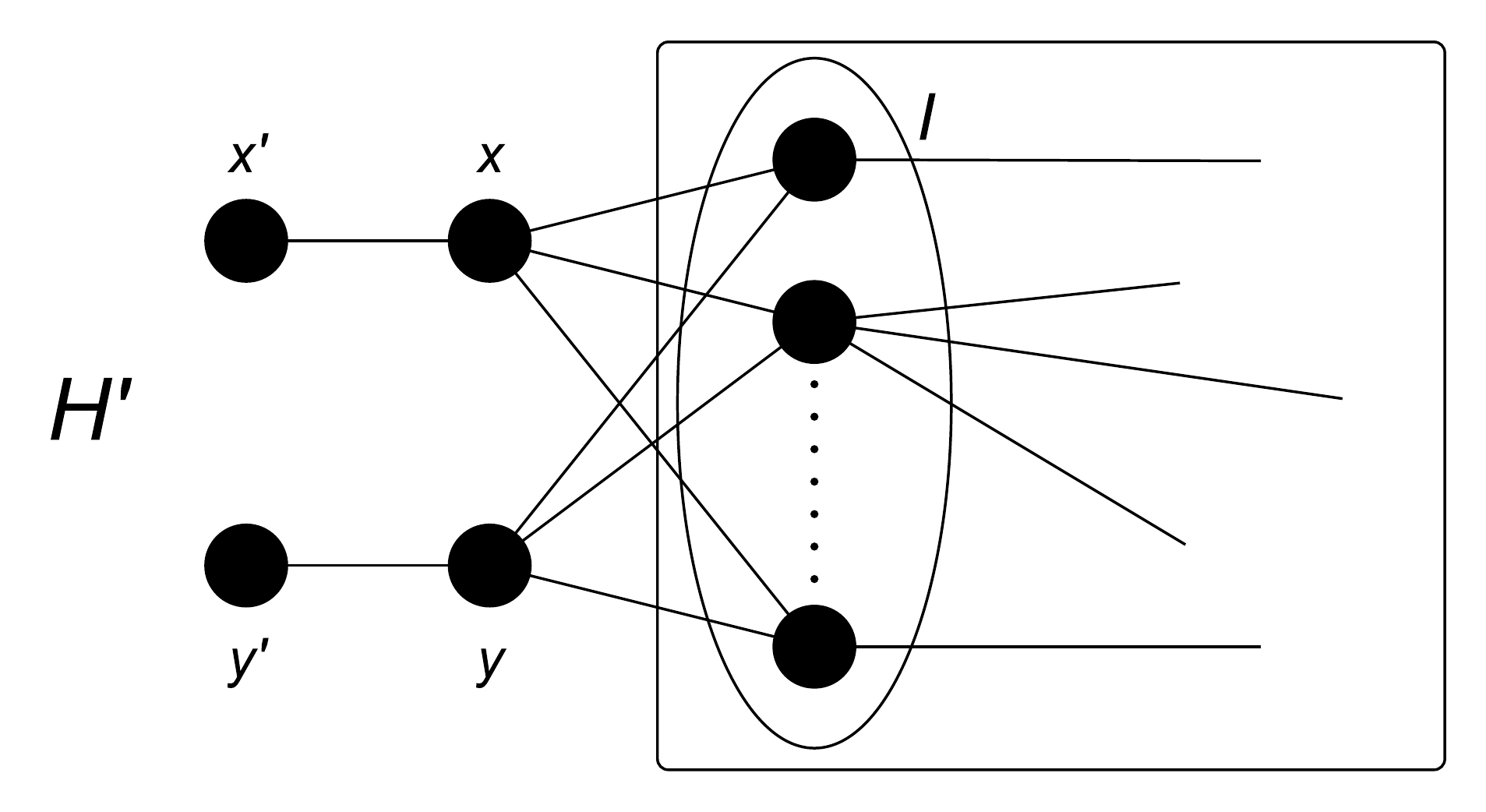}
  \caption{Graph $H'$ constructed from $H$ such that $KB(H') = G- \{q\}$.}
  \label{lemalocografo}
\end{figure}

Finally, if $B=K_{1,2}$ and $H$ is the second graph shown in the statement of Lemma~\ref{lemaLoco}, take $H' = H - \{a\}$. It is clear that the only biclique that is lost
in $H'$ is $B$, while the rest of the structure remains the same. Therefore $KB(H') = G- \{q\}$ where in this case, $H'$ is a subgraph of $H$, as we wanted to prove.  
\qed
\vspace*{0.5cm}

In Figure~\ref{grado2}, we can see an example of application of Theorem~\ref{teoGrado2} using its contrapositive. We want to check if $G$ is not a biclique graph, therefore
we remove one by one the vertices of degree two, $v_1,v_2$ (now with degree two), $v_3$ and $v_4$, and we obtain the graph $G'$ that we know it is not a biclique graph~\cite{marinaYoArxiv}. 
We can conclude then that $G$ and all intermediate graphs, are not biclique graphs as well. 
Note that if we did not know that $G'$ is not a biclique graph and we continued removing the vertices of degree two, we would have
obtained a $K_4$ that actually is a biclique graph therefore we could not conclude anything about $G$ and all intermediate graphs.
Also remark that in all these graphs every induced $P_3$ is contained in an induced \textit{diamond} 
or \textit{gem}, therefore we cannot apply the contrapositive of Theorem~\ref{tMarina}.

\begin{figure}[ht]
  \centering
  \includegraphics[scale=.3]{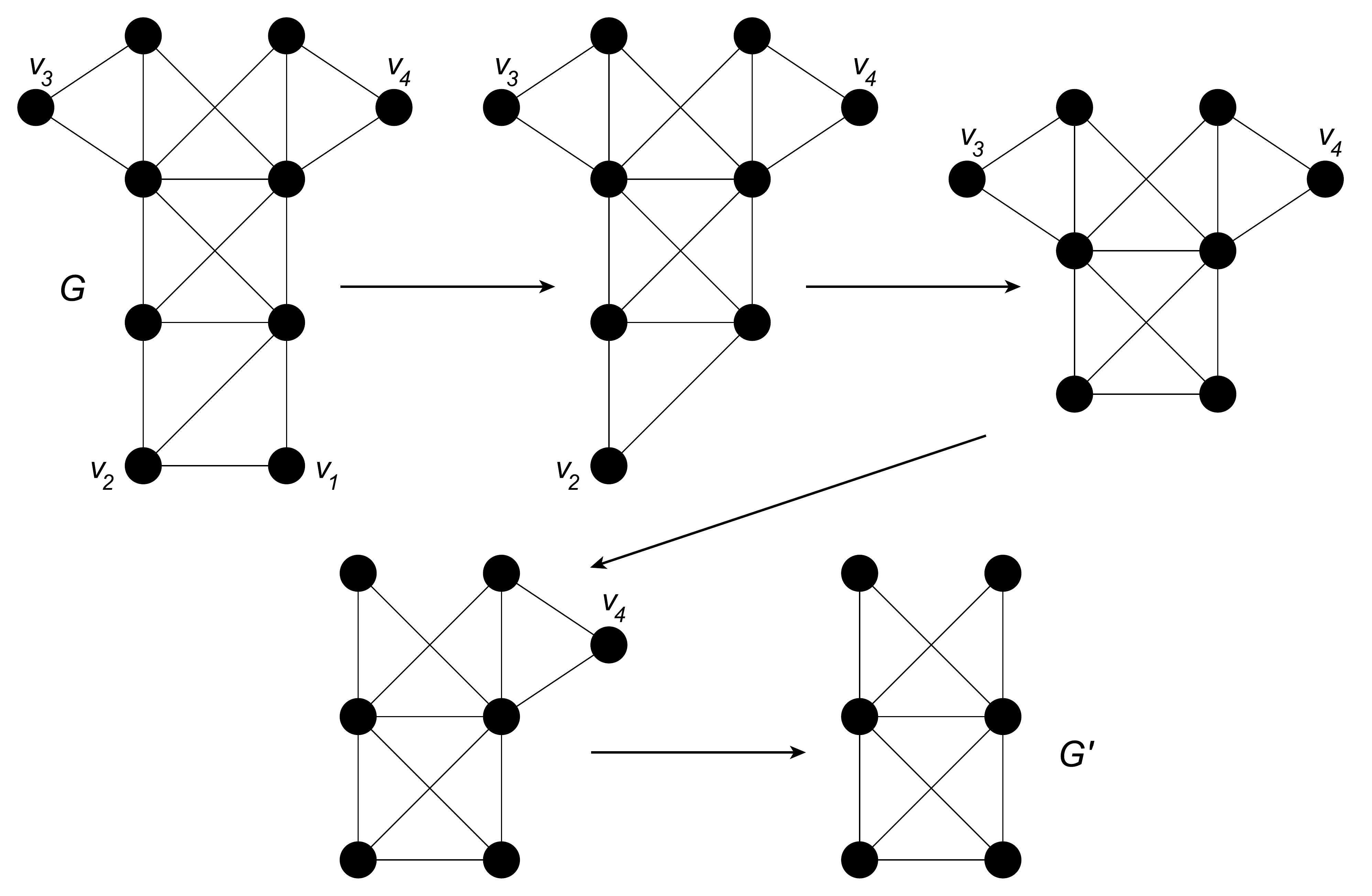}
  \caption{Graph $G$ and all intermediate graphs are not biclique graphs by Theorem~\ref{teoGrado2}.}
  \label{grado2}
\end{figure}

\section{Open problems}\label{conjs}

Our first conjecture states that the family of biclique graphs described in Observation~\ref{obsLean} is actually unique with respect to the property of not having any vertex
such that after its removal, the resulting graph is a biclique graph. 

\begin{conjecture}
Let $G$ be a biclique graph such that $G \neq KB(C_k)$, for $k\geq 7$. Then, there exists a vertex $q \in G$ such that $G-\{q\}$ is a biclique graph.
\end{conjecture}

Note that $G = KB(H) = KB(C_k)$ does not imply that $H=C_k$. Consider the graph $H$ obtained by taking a $C_k$, $k \geq 7$, joining each vertex of the cycle to (zero or more) different new vertices
and finally adding (zero or more) false-twin vertices to the vertices of the cycle. 
It easy to see that $KB(H) = KB(C_k)$ (recall that $KB(H) = KB(Tw(H))$ \cite{marinayo}). 
We propose therefore the following conjecture that characterizes the class of graphs having $KB(C_k)$ as biclique graphs.

\begin{conjecture}
Let $G=KB(H)$ be a biclique graph. Then $G = KB(C_k)$, for $k\geq 7$, if and only if $Tw(H)$ consists of an induced $C_k$ such that each vertex of the cycle is adjacent to at most
one vertex (outside the cycle) that has degree one.
\end{conjecture}

The \textit{if} part is trivial as it was explained right before the statement. If the \textit{only if} part is true, given a graph $H$, one can verify if $KB(H)=KB(C_k)$ in linear time using the modular decomposition given in~\cite{habib} to obtain $Tw(H)$ (then it is easy to check if $Tw(H)$ is an induced $C_k$ with each vertex having at most one neighbor outside the cycle with 
degree one).

To finish, the last conjecture implies the first one when the biclique graphs have false-twin vertices. 

\begin{conjecture}
Let $G$ be a biclique graph and let $q$ be a false-twin vertex of $G$. Then, $G-\{q\}$ is a biclique graph. In particular, $Tw(G)$ is a biclique graph.
\end{conjecture}

\bibliography{biblio}

\end{document}